\documentclass[aps,prl,twocolumn,nofootinbib,longbibliography,superscriptaddress]{revtex4-1}
\usepackage{hyperref}
\usepackage{amsmath}
\usepackage{amssymb}
\usepackage{amsthm}
\usepackage{graphicx}
\newcommand{\be}{\begin{equation}}
\newcommand{\ee}{\end{equation}}
\newcommand{\ran}{\rangle}

\newcommand{\bi}{\begin{itemize}}
\newcommand{\ei}{\end{itemize}}
\newcommand{\mO}{\mathcal{O}}
\newcommand{\CA}{\mathcal{C}_A}
\newcommand{\EA}{\mathcal{E}_A}
\newcommand{\Ab}{\overline{A}}
\newcommand{\ab}{\overline{a}}
\newcommand{\HA}{\mathcal{H}_A}
\newcommand{\HAb}{\mathcal{H}_{\overline{A}}}
\newcommand{\Hc}{\mathcal{H}_{code}}
\newcommand{\Ha}{\mathcal{H}_a}
\newcommand{\Hab}{\mathcal{H}_{\overline{a}}}
\newcommand{\Al}{\mathcal{A}_\text{loc}}
\newcommand{\mH}{\mathcal{H}}
\newcommand{\bfig}{\begin{figure}\begin{center}}
\newcommand{\efig}{\end{center}\end{figure}}
\newtheorem*{theorem}{Theorem}
\def\ba#1\ea{\begin{align}#1\end{align}}
\def\bg#1\eg{\begin{gather}#1\end{gather}}
\def\a{\alpha}

\def\l{\lambda}

\def\p{\phi}
\def\P{\Phi}

\def\r{\rho}
\def\s{\sigma}

\def\y{\psi}

\def\la{\label}

\def\er{\eqref}

\def\eq{\equiv}

\def\qu{\quad}
\def\qqu{\qquad}

\def\({\left(}
\def\){\right)}
\def\[{\left[}
\def\]{\right]}
\def\<{\langle}
\def\>{\rangle}

\def\Tr{{\rm Tr}}
\begin{document}

\title{Reconstruction of Bulk Operators within the Entanglement Wedge
\\in Gauge-Gravity Duality}
\author{Xi Dong}
\affiliation{School of Natural Sciences, Institute for Advanced Study, Princeton, NJ 08540, USA}
\author{Daniel Harlow}
\affiliation{Center for the Fundamental Laws of Nature, Harvard University, Cambridge MA, 02138 USA}
\author{Aron C. Wall}
\affiliation{School of Natural Sciences, Institute for Advanced Study, Princeton, NJ 08540, USA}
\begin{abstract}
In this Letter we prove a simple theorem in quantum information theory, which implies that bulk operators in the Anti-de Sitter / Conformal Field Theory (AdS/CFT) correspondence can be reconstructed as CFT operators in a spatial subregion $A$, provided that they lie in its entanglement wedge.  This is an improvement on existing reconstruction methods, which have at most succeeded in the smaller causal wedge.  The proof is a combination of the recent work of Jafferis, Lewkowycz, Maldacena, and Suh on the quantum relative entropy of a CFT subregion with earlier ideas interpreting the correspondence as a quantum error correcting code.
\end{abstract}

\maketitle
\section{Introduction}
The AdS/CFT correspondence tells us that
certain
large-$N$ strongly-coupled CFTs define theories of quantum gravity in asymptotically AdS space \cite{Maldacena:1997re,Gubser:1998bc,Witten:1998qj,Banks:1998dd,Aharony:1999ti,Heemskerk:2009pn}.  For this definition to be complete, we need to know how to use CFT language to ask any question of interest in the bulk.  The most basic items in this dictionary are:
\bi
\item Quantum states in the Hilbert space of the CFT correspond to quantum states in the bulk.
\item The conformal and global symmetry generators in the CFT correspond to the analogous asymptotic symmetries in the bulk.  For example the CFT Hamiltonian maps to the ADM Hamiltonian, and a $U(1)$ charge maps to the electric flux at infinity.
\item  In the large $N$ limit, any single-trace primary operator $\mO(x)$ in the CFT, with scaling dimension $\Delta$ of order $N^0$, corresponds to a bulk field $\phi(x,r)$.  They are related by the ``extrapolate dictionary'' \cite{Banks:1998dd,Balasubramanian:1998de,Harlow:2011ke}
\be\label{extd}
\mO(x)=\lim_{r\to\infty} r^{\Delta}\phi(x,r) \,.
\ee
\ei
These give an excellent starting point for understanding the correspondence, but the limiting procedure in the extrapolate dictionary makes it difficult to concretely discuss what is going on deep within the bulk.  There is a good reason for this: although the CFT side of the duality is defined exactly, the bulk side is ultimately some nonperturbative theory of quantum gravity, which is only approximately given by a semiclassical path integral over local bulk fields.  This means that any attempt to ``back off of the extrapolate dictionary'' and formulate a CFT representation of the bulk operator $\phi(x,r)$, usually called a reconstruction, must itself be only an approximate notion.  Nonetheless, there is a standard algorithm for producing such a $\phi(x,r)$ perturbatively in $1/N$: one constructs it by solving the bulk equations of motion with boundary conditions \eqref{extd} \cite{Banks:1998dd,Hamilton:2006az,Kabat:2011rz,Heemskerk:2012mn,Kabat:2012hp,Heemskerk:2012np,Morrison:2014jha}.  This algorithm has the nice feature that it can be done quite explicitly, but it has the disadvantage that it relies on solving a rather nonstandard Cauchy problem \cite{Bousso:2012mh,Heemskerk:2012mn}.  In this paper we will refer to it as the HKLL procedure, after \cite{Hamilton:2006az}, who were the first to study it in detail.

The HKLL procedure has the interesting property that it also is sometimes able to reconstruct a bulk operator $\phi(x,r)$ as a CFT operator with nontrivial support only on some spatial subregion $A$ of a boundary Cauchy slice $\Sigma$ \cite{Hamilton:2006az,Morrison:2014jha}.  This is believed to occur whenever the point $(x,r)$ lies in the \textit{causal wedge} of $A$, denoted $\CA$.\footnote{As mentioned above, this reconstruction relies on solving a nonstandard Cauchy problem \cite{Bousso:2012mh,Heemskerk:2012mn}.  In general it is not known if the solution really exists, but it can be found explicitly (in a distributional sense \cite{Morrison:2014jha}) for the case of spherical boundary regions in the vicinity of the AdS vacuum; in this case it is called the AdS-Rindler reconstruction.}  $\CA$ is defined as the intersection in the bulk of the bulk causal future and past of the boundary domain of dependence of $A$ \cite{Hubeny:2012wa}.

Causal wedge reconstruction via the HKLL procedure gives an explicit illustration of ``subregion-subregion duality'', which is the notion that a spatial subregion $A$ in the boundary theory contains complete information about some subregion of the bulk \cite{Bousso:2012sj,Czech:2012bh,Bousso:2012mh}.  It has been proposed however that the subregion dual to $A$ is not just $\mathcal{C}_A$, but rather a larger region: the \textit{entanglement wedge} $\EA$ \cite{Czech:2012bh,Wall:2012uf,Headrick:2014cta}.  To define the entanglement wedge, we must first define the HRT surface $\chi_A$ \cite{Hubeny:2007xt} (the covariant generalization of the Ryu-Takayanagi proposal \cite{Ryu:2006bv}).  $\chi_A$ is defined as a codimension-two bulk surface of extremal area, which has boundary $\partial \chi_A=\partial A$ and is homologous to $A$ through the bulk (in the event that multiple such surfaces exist, it is the one with least area).  The HRT formula tells us that at leading order in $1/N$, the area of $\chi_A$ is proportional to the von Neumann entropy of the region $A$ in the CFT \cite{Ryu:2006bv,Hubeny:2007xt,Lewkowycz:2013nqa}.  The entanglement wedge $\EA$ is then defined as the bulk domain of dependence of any achronal bulk surface whose boundary is $A \cup \chi_A$.

Recently, plausibility arguments for entanglement wedge reconstruction have been developing along two different lines. \cite{Almheiri:2014lwa} proposed a re-interpretation of causal wedge reconstruction as quantum error correction, which is a structure that very naturally allows an extension to entanglement wedge reconstruction.  This was explicitly realized in toy models \cite{Pastawski:2015qua,Hayden:2016cfa}.  At the same time it has gradually been understood that a proposal \cite{Faulkner:2013ana} for the first $1/N$ correction to the HRT formula is indicative that entanglement wedge reconstruction should be possible \cite{Almheiri:2014lwa,Jafferis:2014lza,Jafferis:2015del}.  In particular Jafferis, Lewkowycz, Maldacena, and Suh (JLMS) have given a remarkable bulk formula for the modular Hamiltonian associated to any CFT region, which implies that the relative entropy of two states in a boundary spatial region $A$ is simply equal to the relative entropy of the two bulk states in $\EA$ \cite{Jafferis:2015del}.  The purpose of this paper is to tie all of these ideas together into a proof that entanglement wedge reconstruction is in fact possible in AdS/CFT: given a boundary subregion $A$, all bulk operators in $\EA$ have CFT reconstructions as operators in $A$.

\section{Entanglement Wedge Reconstruction as Quantum Error Correction}
We mentioned above that the holographic nature of AdS/CFT prevents bulk operator reconstruction from being a precise notion.  This reveals itself not just in the perturbative nature of the the HKLL algorithm, but also in its regime of validity.  As one acts with more and more reconstructed operators in a region, the threshold for black hole formation is eventually crossed.  This leads to a breakdown of the construction.  In \cite{Almheiri:2014lwa} this was formalized into the notion that one should think of reconstructed bulk operators as only making sense in a \textit{code subspace} of the CFT Hilbert space, which we denote $\Hc$.  The choice of this subspace is not unique, since in general we can choose to define it based on whatever observables we are interested in studying, but the simplest thing to do is choose a state which we know has a geometric interpretation, and then consider the subspace of all states where the backreaction of the metric about that geometry is perturbatively small.\footnote{This is an overly conservative definition of the code subspace, since a certain amount of backreaction (such as that present in the solar system) can be included by resumming subclasses of diagrams in the perturbative expansion \cite{Heemskerk:2012mn,Marolf:2015dia}.  We adopt it nonetheless to simplify our arguments below.}

As a concrete example, we can consider the subspace of states of the CFT on a sphere whose energy is less than that of a Planck-sized black hole in the center of the bulk, together with the image of this subspace under the conformal group.  For example in the $\mathcal{N}=4$ super Yang-Mills theory with gauge group $SU(N)$, and with gauge coupling $g\sim 1$ to equate the string and Planck scales, this corresponds to the set of primary operators whose dimensions are $\lesssim N^{1/4}$, together with their conformal descendants.  If the duality is correct, meaning that correlation functions of local boundary operators in these states can be reproduced by bulk Feynman-Witten diagrams with some effective action,\footnote{See \cite{Heemskerk:2009pn} for a discussion of which CFTs have this property, and also for how to determine the effective action in terms of CFT data if they do.} then the machinery of \cite{Banks:1998dd,Hamilton:2006az,Kabat:2011rz,Heemskerk:2012mn,Kabat:2012hp,Heemskerk:2012np,Morrison:2014jha} enables us to explicitly reconstruct all low-energy bulk operators $\phi(x,r)$ in such a way that their correlation functions in any state in the code subspace agree with those computed in bulk effective field theory with that effective action to all orders in $1/N$.  For this to work in detail we must choose some sort of covariant UV cutoff in the bulk, but we will not discuss this explicitly in what follows.

We can now describe entanglement wedge reconstruction.  Say that we split a Cauchy slice $\Sigma$ of the boundary CFT into a region $A$ and its complement $\Ab$.  The CFT Hilbert space has a tensor factorization as
\be
\mathcal{H}_{CFT}=\HA\otimes \HAb \,.
\ee
Similarly we can think of the code subspace as factorizing as
\be\label{codefactor}
\Hc=\Ha\otimes \Hab \,,
\ee
where $\Ha$ denotes the Hilbert space of bulk excitations in $\EA$, and $\Hab$ denotes the Hilbert space of bulk excitations in $\mathcal{E}_{\overline{A}}$.  We illustrate this in Figure \ref{fig}.
\bfig
\includegraphics[height=3cm]{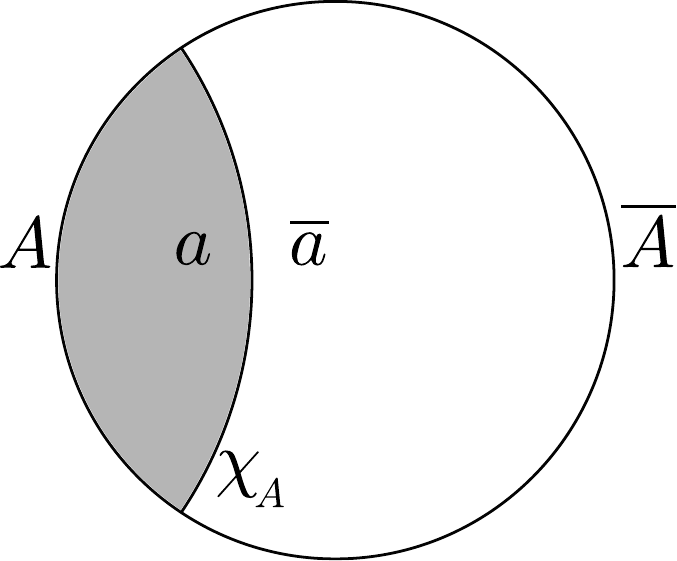}
\caption{Factorizing the bulk and boundary on a time slice.  The entanglement wedge $\EA$ is shaded.  For simplicity we have shown a connected boundary region $A$, although this might not be the case. }\label{fig}
\efig
Both tensor factorizations are complicated by the fact that gauge constraints might be present \cite{Donnelly:2011hn,Casini:2013rba,Donnelly:2014fua,Donnelly:2015hxa,Radicevic:2015sza,Soni:2015yga,Ma:2015xes,Donnelly:2016auv}, but likely this problem can be absorbed into the choice of UV cutoff (as shown for $U(1)$ gauge fields by \cite{Harlow:2015lma}).  We could dispense with this issue by formulating our arguments in terms of subalgebras instead of subfactors, see \cite{Casini:2013rba} for the relevant definitions, but we have opted for the latter for familiarity.

Entanglement wedge reconstruction is then the statement that any bulk operator $O_a$ acting within $\Ha$ can always be represented in the CFT with an operator $O_A$ that has support only on $\HA$.  In fact this is the precise definition of the idea that the operator $O_a$ can be ``corrected'' for the erasure of the region $\Ab$ \cite{beny2007generalization}, and the observation that a given $O_a$ can be reconstructed on different choices of $A$ reflects the ability of the code to correct for a variety of erasures \cite{Almheiri:2014lwa}.  We will establish the existence of $O_A$ below, but first we need to review a recent result that will be essential for the proof.

\section{Review of the JLMS Argument}
In \cite{Jafferis:2015del}, JLMS argued for an equivalence of relative entropy between bulk and boundary.  In this section we give a more detailed derivation of this result, which also extends it to higher orders in the semiclassical expansion given an assumption we state precisely below. Let us first recall that the von Neumann entropy of any state is defined as
\be
S(\rho)\equiv -\Tr \left(\rho \log \rho \right) \,,
\ee
its modular Hamiltonian is defined as
\be
K_\rho\equiv -\log \rho \,,
\ee
and the relative entropy of a state $\rho$ to a state $\sigma$ is
\be\label{relent}
S(\rho|\sigma)\equiv \Tr \left(\rho \log \rho\right)-\Tr\left(\rho \log \sigma\right) =-S(\rho)+\Tr \left(\rho K_\sigma\right) \,.
\ee
Relative entropy is non-negative, and vanishes if and only if $\rho=\sigma$.  Under small perturbations of the state, the entropy and modular Hamiltonian are related by the ``first law of entanglement'':\footnote{This is derived in e.g.~\cite{Bianchi:2012br,Lashkari:2013koa} by linearizing the entropy.  To deal with possible non-commutativity of $\rho$ and $\delta\rho$ one can use the Baker-Campbell-Hausdorff formula and recall that $\Tr\left(A[B,C]\right)=0$ if $A$ and $B$ are simultaneously diagonalizable.}
\be\label{entfl}
S(\rho+\delta\rho)-S(\rho)=\Tr \left(\delta \rho\, K_\rho\right)+\mO(\delta\rho^2) \,.
\ee

Now consider the setup of the previous section, where in some holographic CFT we pick a code subspace $\Hc=\Ha\otimes \Hab$ in which effective field theory perturbatively coupled to gravity is valid, and in which all states can be derived from path integral constructions in this effective theory.  Following JLMS, we recall that Faulkner et al. \cite{Faulkner:2013ana} showed that such states\footnote{The argument of \cite{Faulkner:2013ana}, like the earlier classical argument \cite{Lewkowycz:2013nqa}, assumes that the von Neumann entropy can be found by an analytic continuation of certain ``Renyi entropies'' $\Tr(\rho^n)$, and that the bulk path integral for calculating the $n$-th Renyi entropy does not spontaneously break the associated $\mathbb{Z}_n$ symmetry.  There is also some subtlety in applying the argument to states that do not possess a moment of time-reflection symmetry, but this seems to just be a matter of convenience, and an argument that dispenses with this criterion will appear elsewhere soon \cite{dlr}.} $\rho$ obey
\be\label{flm}
S(\rho_A)=S(\rho_a)+\Tr\left(\rho_a \Al\right) \,.
\ee
Here $\Al$ denotes a bulk operator that is a local integral over the HRT surface $\chi_A$.  At leading order in $1/N$, or equivalently at leading order in the gravitational coupling $G$, we have
\be
\Al = \frac{\text{Area}(\chi_A)}{4G} \,,
\ee
as required by the HRT formula, but $\Al$ receives corrections at higher orders in $1/N$ \cite{Barrella:2013wja,Faulkner:2013ana}, or in the presence of more general gravitational interactions \cite{Wald:1993nt,Iyer:1994ys,Jacobson:1993vj,Jacobson:1993xs,Solodukhin:2008dh,Hung:2011xb,Bhattacharyya:2013jma,Fursaev:2013fta,Dong:2013qoa,Camps:2013zua,Miao:2014nxa}.

In fact \cite{Faulkner:2013ana} only established \eqref{flm} to order $N^0$.  In \cite{Engelhardt:2014gca} it was suggested that \eqref{flm} continues to hold to all orders in $1/N$, provided that one always defines $\chi_A$ to be extremal with respect to the sum on the right hand side of \eqref{flm} (sometimes called the ``generalized entropy''); such a $\chi_A$ is known as a quantum extremal surface.  We will assume this to be the case throughout the code subspace in what follows: if it is not then our arguments only establish entanglement wedge reconstruction to order $N^0$.

To establish the JLMS result, we now observe that linearizing \eqref{flm} about $\sigma$ and using the first law \eqref{entfl}, we have
\be\label{inf}
\Tr\left(\delta\sigma_A K_{\sigma_A}\right)=\Tr\left(\delta\sigma_a\left(\Al^{\{\sigma\}}+K_{\sigma_a}\right)\right) \,.
\ee
Here we have taken $\delta \sigma$ to be an arbitrary perturbation that acts within the code subspace $\Hc$.  We have written $\Al^{\{\sigma\}}$ to emphasize that $\Al$ is still located at the surface defined by extremizing $S(\sigma_a)+\Al$.  \eqref{inf} is linear in $\delta \sigma$, so we can integrate it to obtain
\be\label{full}
\Tr\left(\rho_A K_{\sigma_A}\right)=\Tr\left(\rho_a^{\{\sigma\}}\left(\Al^{\{\sigma\}}+K_{\sigma_a}\right)\right) \,,
\ee
where now $\rho$ and $\sigma$ are arbitrary states acting within $\Hc$.  We have written $\rho_a^{\{\sigma\}}$ to clarify that we are still factorizing the bulk Hilbert space at the quantum extremal surface for $\sigma$, even though we are now considering another state $\rho$ as well.  Since \eqref{full} holds for any $\rho$ acting within $\Hc$, it implies that
\be\label{JLMS2}
\Pi_c K_{\sigma_A}\Pi_c= K_{\sigma_a}+\Al^{\{\sigma\}} \,,
\ee
where $\Pi_c$ is the projection operator onto the code subspace, and where we have defined the operators on the right hand side to annihilate $\Hc^\perp$.  This is one of the main results of \cite{Jafferis:2015del}.

Moreover combining \eqref{relent}, \eqref{flm}, and \eqref{full}, we find that
\begin{align}\nonumber
S(\rho_A|\sigma_A)=&S(\rho_a^{\{\sigma\}}|\sigma_a)\\\nonumber
&+\Big[\Tr\left(\rho_a^{\{\sigma\}}\Al^{\{\sigma\}}\right)-\Tr\left(\rho_a^{\{\rho\}}\Al^{\{\rho\}}\right)\\
&+S\left(\rho_a^{\{\sigma\}}\right)-S\left(\rho_a^{\{\rho\}}\right)\Big] \,.\label{JLMS}
\end{align}
To order $N^0$ the terms in square brackets cancel
due to the extremality of $\chi_A^{\{\rho\}}$ with respect to $S(\rho_a) + \Al$,
so we find the equivalence of relative entropies which is the other main result of \cite{Jafferis:2015del}.\footnote{This equivalence is complicated by gravitons, which are metric variations of order $\sqrt{G}$ and thus can produce second-order variations of order $N^0$.  These can be dispensed with by introducing a gauge (e.g. the one in \cite{Jafferis:2015del}) where the quantum extremal surface $\chi_A$ does not move until order $N^0$.}

More generally, we find that the relative entropies differ by the difference of generalized entropies of $\rho$ on the two quantum extremal surfaces $\chi_A^{\{\rho\}}$ and $\chi_A^{\{\sigma\}}$.  In either case, we observe that if $\rho_a\!{}^{\{\sigma\}}=\sigma_a$, then \eqref{JLMS} implies $S(\rho_A|\sigma_A)=0$, and thus $\rho_A=\sigma_A$.  This is the result that we will use in the proof below.

By symmetry, the results \eqref{JLMS},\eqref{JLMS2} apply also for the complement regions $\Ab$ and $\ab$, although for a mixed state one must use $S(\rho_{\ab})$ rather than $S(\rho_{a})$ when extremizing to find $\chi_{\Ab}$.\footnote{In this case, we expect that the two entanglement wedges can never overlap.}

At order $N^0$ it is intuitively clear that \eqref{JLMS} says that we must be able to reconstruct in the entanglement wedge.  Relative entropy is a measure of the distinguishability of two quantum states, so we can only have \eqref{JLMS} if $\rho_A,\sigma_A$ have just as much information about the bulk as $\rho_a,\sigma_a$.  We now prove a theorem that makes this precise, and extends it to higher orders in $1/N$ given the proposal of \cite{Engelhardt:2014gca}.

\section{A Reconstruction Theorem}
\begin{theorem}
Let $\mH$ be a finite-dimensional Hilbert space, $\mH = \HA \otimes \HAb$ be a tensor factorization, and $\Hc$ be a subspace of $\mH$.  Let $O$ be an operator that, together with its Hermitian conjugate, acts within $\Hc$.  If for any two pure states $|\p\>, |\y\> \in \Hc$, there exists a tensor factorization $\Hc = \Ha \otimes \Hab$ such that $O$ acts only on $\Ha$, and the reduced density matrices
\ba\la{rdm}\nonumber
\r_{\Ab} &\eq \Tr_A |\p\> \<\p| \,,\qqu
\s_{\Ab} \eq \Tr_A |\y\> \<\y| \,,\\
\r_{\ab} &\eq \Tr_a |\p\> \<\p| \,,\qqu
\s_{\ab} \eq \Tr_a |\y\> \<\y|
\ea
satisfy
\be\la{cond}
\rho_{\ab}=\sigma_{\ab} \qu\Rightarrow\qu \rho_{\Ab}=\sigma_{\Ab} \,,
\ee
then both of the following statements are true:
\begin{enumerate}
\item For any $X_{\Ab}$ acting on $\HAb$ and any state $|\p\> \in \Hc$, we have
\be\la{zcomm}
\<\p| [O, X_{\Ab}] |\p\> =0 \,.
\ee
\item There exists an operator $O_A$ acting just on $\HA$ such that $O_A$ and $O$ have the same action on $\Hc$, i.e.
\be\la{idact}
O_A |\p\> = O |\p\> \,,\qqu
O_A^\dag |\p\> = O^\dag |\p\> \,,
\ee
for any state $|\p\> \in \Hc$.
\end{enumerate}
\end{theorem}
\begin{proof}
First we note that the two statements are guaranteed to be equivalent by the theorem proved in Appendix B of \cite{Almheiri:2014lwa}; for the convenience of the reader we sketch the logic of that proof in an appendix below.  Therefore we will only need to prove the first statement.  Furthermore it is sufficient to prove the theorem when $O$ is a Hermitian operator which we assume now.\footnote{To see this, we recall that any operator is a (complex) linear combination of two Hermitian operators.}

Consider any state $|\p\> \in \Hc$, and let $\l$ be an arbitrary real number.  For the two states $|\p\> $ and
\be\la{psidef}
|\y\> \eq e^{i\l O} |\p\>
\ee
which are both in $\Hc$, the assumption of the theorem guarantees the existence of a tensor factorization $\Hc = \Ha \otimes \Hab$ such that $O$ acts only on $\Ha$, which implies
\be
\r_{\ab} = \s_{\ab}
\ee
because the two states only differ by the action of a unitary operator $e^{i\l O}$ on $\Ha$.  Using \er{cond} we find
\be\la{sdm}
\r_{\Ab} = \s_{\Ab} \qu\Rightarrow\qu \<\y |X_{\Ab}| \y\> = \<\p |X_{\Ab}| \p\> \,.
\ee
Using \er{psidef} we may rewrite \er{sdm} as
\be
\<\p | e^{-i\l O} X_{\Ab} e^{i\l O} | \p\> - \<\p |X_{\Ab}| \p\> =0 \,.
\ee
Expanding this equation to linear order in $\l$, we find
\be
\<\p | [O, X_{\Ab}] | \p\> =0 \,,
\ee
which proves \er{zcomm}, and hence also \er{idact}.
\end{proof}

In the above theorem we assumed that $\mH$ is finite-dimensional to avoid any subtleties with the proof of Appendix B of \cite{Almheiri:2014lwa}.  In AdS/CFT we can accomplish this by introducing a UV cutoff in the CFT. This only affects physics near the asymptotic boundary of AdS, and therefore is not essential to our discussion.  In bulk language the assumptions for the theorem follow from picking an $O$ that lies in the entanglement wedge $\EA$ for any state in $\Hc$, with the choice of factorization coming from the quantum extremal surface for the state $|\phi\ran$.  This ensures that we can apply eq. \eqref{JLMS} for the complement region $\Ab$, which then implies \eqref{cond}.

\section{Discussion}
We view the theorem of the previous section as establishing the existence of entanglement wedge reconstruction.  Several comments are in order:

Our argument ultimately rests on the validity of the quantum version \eqref{flm} of the HRT formula, derived to order $N^0$ by \cite{Faulkner:2013ana}, and conjecturally extended to higher orders in $1/N$ by \cite{Engelhardt:2014gca}.  Understanding those results better is thus clearly of interest for bulk reconstruction.  Any discussion of bulk reconstruction will ultimately be approximate, so it would be interesting to prove an approximate version of our theorem and use it to clarify the stability of our results under small perturbations.

Note that our proof is constructive.  After establishing the assumptions of the theorem of Appendix B in \cite{Almheiri:2014lwa} hold, that theorem provides an explicit formula for the reconstruction $O_A$.  Our construction is thus an improvement on the HKLL procedure even within the causal wedge, since the nonstandard Cauchy problem involved in causal wedge reconstruction has only been solved in special cases.  However, it does not give much insight into how to think about the reconstruction from a bulk point of view.  This is to be contrasted with the HKLL algorithm, which (at least when it works) proceeds by solving bulk equations of motion.  We believe that a bulk interpretation should exist, and finding it could be quite illuminating.

\section*{acknowledgments}
{\small This work was initiated at the Quantum Error Correction and Tensor Network Project Meeting of the It from Qubit Collaboration.  We would like to thank the participants for many useful discussions, the KITP for hospitality during the meeting, the Simons foundation for funding the collaboration, and the National Science Foundation for supporting the KITP under Grant No. NSF PHY11-25915.  We would also like to thank Aitor Lewkowycz for useful discussions about gravitons and the extremal surface, and William Donnelly for pointing out a defect in an earlier version of our proof of the theorem of Section 4.  XD is supported in part by the Department of Energy under grant DE-SC0009988 and by a Zurich Financial Services Membership at the Institute for Advanced Study.  DH is supported by DOE grant  DE-FG0291ER-40654 and the Harvard Center for the Fundamental Laws of Nature. AW is supported by the Martin A. \& Helen Chooljian Membership at the Institute for Advanced Study.}

\appendix\label{proof}
\section{Appendix: Theorem of Operator Algebra Quantum Error Correction}

In this appendix, we give an intuitive summary of the derivation of the theorem establishing the equivalence between \er{zcomm} and \er{idact}.  More details can be found in the full proof provided in Appendix B of \cite{Almheiri:2014lwa}.

Suppose that $\Hc$ is spanned by an orthonormal basis $|i\>_{A\Ab}$.  Consider the state
\be\la{pref}
|\P\> = \sum_i |i\>_R \otimes |i\>_{A\Ab} \,,
\ee
where $R$ is a reference system whose Hilbert space is spanned by an orthonormal basis $|i\>_R$.  Using \er{pref} we may mirror any operator $O$ acting within $\Hc$ to an operator $O_R$ on $R$ such that $O|\P\> = O_R|\P\>$, $O^\dag |\P\> = O_R^\dag |\P\>$.

The task is then to instead view $O_R$ as $O_R \otimes I_{\Ab}$ and mirror it back to an operator $O_A$ on $A$.  To do this we need the Schmidt decomposition
\be
|\P\> = \sum_\a c_\a |\a\>_A \otimes |\a\>_{R\Ab} \,,
\ee
where the states $|\a\>_{R\Ab}$ with $c_\a \neq 0$ generally span a subspace of $\mathcal{H}_R \otimes \HAb$.  We can mirror $O_R \otimes I_{\Ab}$ onto $A$ if it acts within this subspace which is guaranteed by
\be
[O_R \otimes I_{\Ab}, \r_{R\Ab}(\P)] = 0 \,.
\ee
This statement is implied by \er{zcomm}.  Therefore we obtain \er{idact}.

\bibliography{bibliography}
\end{document}